\documentclass{amsart}
\usepackage{amsfonts}
\usepackage{amsmath}
\usepackage{amssymb}
\usepackage{amsthm}
\usepackage{color}
\usepackage{comment}
\usepackage{natbib}
\usepackage[foot]{amsaddr}

\usepackage{pdfpages}

\definecolor{Green}{rgb}{0.20,0.43,0.09}

\newtheorem{Thm}{Theorem}
\newtheorem{Lem}[Thm]{Lemma}

\theoremstyle{definition}
\newtheorem{Def}[Thm]{Definition}

\theoremstyle{remark}

\numberwithin{equation}{section}

\def\RR{{\mathbb{R}}}
\def\ZZ{{\mathbb{Z}}}

\title[Twin Paradox and the No Show Paradox]{Topological aggregation,\\
the Twin Paradox and the No Show Paradox}

\date{\today}

\author[G.~Ch\`eze]{Guillaume Ch\`eze}
\address{Guillaume Ch\`eze: Institut de Math\'ematiques de Toulouse\\
Universit\'e Paul Sabatier \\
118 route de Narbonne\\
31 062 TOULOUSE cedex 9, France}
\email{guillaume.cheze@math.univ-toulouse.fr}

\date{\today}

\begin{document}
	
\begin{abstract}
Consider the framework of topological aggregation introduced by \cite{Chichilnisky80}. We prove that in this framework the Twin Paradox and the No Show Paradox cannot be avoided. Anonymity and unanimity are not needed to obtain these results.
\end{abstract}	
	
\maketitle


\section*{Introduction}
The goal of this note is to show that we cannot avoid some paradoxes in the topological social choice setting. Topological social choice has been introduced by \cite{Chichilnisky80}. Nowadays, there exists different surveys and collections of papers about this topic, e.g. \citep{Baigent,Heallivre,Lauwers2000,Lauwers,Mehta} and roughly speaking it is the study of continuous aggregation rules. 

In the topological social choice framework, an aggregation rule between $k$ agents is  a function from $X^k$ to $X$, where $X$ is a topological space. Chichilnisky has  shown an impossibility result:  there exists no continuous aggregation rule defined on the sphere $S^n$ satisfying anonymity and unanimity. This result has been proved independently by \cite{Eckmann1954,Eckmann2004}. \\

In this note we study the Twin Condition (roughly, identical twins have more power) and the Participation Condition (voting does not harm). With an homological approach, we prove  that on the sphere $S^n$  neither the Twin Condition nor the Participation Condition can be satisfied by a continuous aggregation rule.  As the Twin Condition (respectively the Participation Condition) is not satisfied for continuous aggregation rule on the sphere we say that we have a Twin Paradox (respectively a No Show Paradox).\\

When we study the Twin Condition, two players have exactly the same preference. We can think about this situation in the following way: one player moves his/her preference to a new preference $x$. Furthermore, this new preference $x$ was already the preference of another player. Thus $x$ is now the choice of more players. Therefore, we hope that the new social preference will reflect this  situation and will be closer to $x$.\\
There already exists papers studying the behavior of aggregation rules when players preferences move closer to the preference of another player. \\
 \cite{Macintyre} has used a convexity condition and a monotonicity condition in order to study this kind of situations for the circle $S^1$. For two players with different preferences the convexity condition used by Macintyre is roughly the following: the social outcome must be contained in the small semicircle with end-points the preference of these two players. The monotonicity condition means that if the first player moves his/her preference closer to the preference of the second player then the social outcome must not penalize the second player.
With elementary methods and no algebraic topology Macintyre has shown that the only possible continuous aggregation rules with domain $S^1$ satisfying his convexity condition and the monotonicity condition are dictatorial.  This result is given for two players with preference in $S^1$. In the same article Macintyre has also studied the situation with $k$ players with preferences in $S^1$. In this case the convexity condition becomes: the social outcome must be contained in a semicircle containing a \emph{majority} of players with end-points the preference of at least two of the players. Macintyre has then shown that there exists no continuous aggregation rule satisfying the convexity condition for $k\geq 2$ players. This result is still proved with elementary methods and only for the circle $S^1$.  \\
\cite{Prox-cond} have  studied the role of what they called a proximity condition for merging function. A merging function is a function $f:\RR^k \rightarrow \RR$, where $\RR$ denotes the real line. The role of a merging function is to aggregate real numbers $x_1,\ldots, x_k$ given by the players to a  single value $f(x_1,\ldots,x_k)$.  For example, the arithmetic mean is a merging function. The authors have studied merging function satisfying the following proximity condition: suppose that \emph{every} $x_j$ moves strictly closer to some individual's number $x_i$ then the merging function does not change or move closer to the individual's number $x_i$. In this context, closer means closer relatively to the usual distance on  real numbers. Duddy and Piggins have shown that scale invariant functions satisfying their proximity condition are dictatorial. \\

Here we study functions with domain $S^n$, with $n \geq 1$. Our Twin Condition is in the same vein than the conditions given by Macintyre, Duddy and Piggins. However  we prove that there exists no aggregation rule satisfying our condition. Furthermore, contrarily to Macintyre, Duddy and Piggins we do not need to consider moves for every players or for a majority of voters, a condition about two players is sufficient to get our impossibility result.\\

In the discrete setting, the No Show Paradox has been introduced by \cite{FishburnBrams}. This paradox means that there exists a player who would rather not vote.  Sometimes, this paradox is also called the Abstention Paradox.\\
\cite{Moulin} has proved that if a  voting function chooses the Condorcet winner when it exists, then  this voting rule generates a No Show Paradox whenever there are at least 4 alternatives and  25 voters. We recall that outcomes of voting functions are singleton sets.\\
\cite{Perez} has studied a stronger version of the No Show paradox. This stronger version says that there exists a player whose favorite candidate loses the election if he/she votes honestly, but gets elected if he/she abstains. Perez has shown that almost all known voting rules suffer from this paradox.\\
\cite{JimenoPerez} has extended Moulin's result to Condorcet voting correspondences. In this case, outcomes for the voting rule are nonempty sets and not singleton sets.\\
Recently, \cite{AAMAS2016} has shown that Moulin's result is still valid if there are at least 12 voters. Furthermore, this bound is optimal since the authors have given an example of a voting rule satisfying the Condorcet criterion and avoiding the No Show Paradox for up to 11 voters.\\

To the author's knowledge, the No Show paradox has not been studied in the topological social choice framework. However, we can mention that there exists relations between the "classical" social choice setting and the topological social choice setting. What we call "classical" social choice is the study of aggregation rules in a discrete setting (with the discrete topology). The relation between this two setting has been given by \cite{Baryshnikov}. Baryshnikov and then \cite{Tanaka} have given a proof of Arrow's theorem using  an homological strategy  closed to the one given by Chichilnisky in her study of continuous aggregation rule over the sphere $S^n$. In this paper, we will also used an homological approach to prove our theorems.  \\

The structure of this note is the following: In the first section we study the Twin Paradox, in the second section we study the No Show Paradox, and in an appendix we recall all the  basic notions of algebraic topology used during our proofs.

\section{The twin paradox}
In this section we define and study the Twin Paradox. \\
Suppose that an aggregation rule in a topological space $X$ takes individual preferences $x_1, \ldots, x_k$ and gives the social preference $x $. Then, suppose that the $j$-th player changes his point of view and decide to modify his choice from $x_j$ to $x_i$. Then the $i$-th player and the $j$-th player have exactly the same preference and we can consider the $j$-th player as a twin of the $i$-th player. With this new situation, we can hope that the aggregation rule will give a new social preference closer to $x_i$, if it is not the case we have a Twin paradox. Indeed, $x_i$ is now the choice of more players and then we hope that the new social preference will reflect this  situation and will be closer to $x_i$.  \\
Here closer means that we consider $x_i$ in a metric space. We  now define the Twin Condition which corresponds to the situation explained above.

\begin{Def}
Consider a metric space $(X,d)$ where $d$ is the metric on $X$. We say that $f:X^k\rightarrow X$ satisfies the \emph{Twin Condition}  if  for all $i \neq j$ and all $x_1,\ldots, x_k \in X$ with $x_i\neq x_j$ we have:
$$d\Big(f(x_1,\ldots,x_i, \ldots,x_{j-1}, x_i,x_{j+1}, \ldots, x_k), \,x_i\Big) \leq d\Big(f(x_1,\ldots, x_k), \, x_i\Big),$$
and the inequality is strict for  $x_i \neq f(x_1,\ldots, x_k)$.
\end{Def}

The Twin Paradox is a situation where the Twin Condition is not satisfied.\\

In the following, we study the Twin Condition on the sphere $S^n$. The distance between two points on the sphere is the length of a great circle joining these two points. If $x, y \in S^n$ then we denote by $d_{S^n}(x,y)$  the distance on the sphere between these two points. We have $d_{S^n}(x,y)=\arccos(x\cdot y)$, where $x\cdot y$ is the usual scalar product of $x, y \in \RR^{n+1}$. We can already remark that $d_{S^n}(x,y)\leq \pi$ for all $x,y \in S^n$ and $d_{S^n}(x,y)=\pi$ means $x=-y$.\\

The main result of this section is the following:
\begin{Thm}\label{Thm:twin}
For $k\geq 3$ players, there exists no continuous aggregation rule \\ \mbox{$f:(S^n)^k \rightarrow S^n$} which satisfies the Twin Condition.
\end{Thm}

 We remark  that a dictatorship do not satisfy the Twin Condition, but the non-existence of an aggregation rule satisfying the Twin Condition is not straightforward.\\

In order to prove this theorem we first prove  that the Twin Condition on the sphere entails the following condition:
\begin{Def}
We say that $f:X^k \rightarrow X$ satisfies the \emph{Nowhere Anti-Unanimity Condition} if for all $x \in X$ we have $f(x,x,\ldots,x) \neq -x$.
\end{Def}
Remark that this definition is still valid when $k=1$. This just means that the function satisfies $f(x) \neq -x$. Thus in the following we still use the term Nowhere Anti-Unanimity Condition even if $k=1$.\\

When the space $X$ is the sphere  $S^n$ with the distance $d_{S^n}$, we have:

\begin{Lem}\label{twinanti}
Suppose that  $f:(S^n)^k \rightarrow S^n$ satisfies the Twin condition then $f$ satisfies the Nowhere  Anti-Unaninimity Condition.
\end{Lem}

\begin{proof}
By contradiction. Assume the existence of a point $x_0 \in S^n$ such that $f(x_0,\ldots,x_0) =-x_0$. As $f$ satisfies the Twin Condition we have for all $y \in S^n$ such that $y \neq x_0$: 
 $$\pi=d_{S^n}(-x_0,x_0)=d_{S^n}\big(f(x_0,\ldots,x_0),x_0\big)\leq  d_{S^n}\big(f(x_0,\ldots,x_0, y, x_0,\ldots,x_0),x_0\big).$$
 Then $f(x_0,\ldots,x_0,y,x_0,\ldots,x_0\big)=-x_0$, thus $x_0 \neq f(x_0,\ldots,x_0,y,x_0,\ldots,x_0\big)$. In this situation the Twin Condition implies that the previous inequality is strict.
 As the distance between two points on the sphere is never strictly bigger than than $\pi$, we get a contradiction and then we deduce that for all $x \in S^n$ we have $f(x,\ldots,x) \neq -x$.
\end{proof}

Now, we introduce some notations.

\begin{Def}
We denote by $\delta^{(k)}_{i,j}$ the following function:
\begin{eqnarray*}
\delta^{(k)}_{i,j}: S^n & \longrightarrow & (S^n)^k\\
x&\longmapsto & (e_1,\ldots,e_1,x,e_1,\ldots,e_1,x,e_1,\ldots,e_1)
\end{eqnarray*}
where $i \neq j $, $e_1=(1,0,\ldots,0) \in S^n \subset \RR^{n+1}$ and $x$ appears in $\delta^{(k)}_{i,j}(x)$ in the $i$-th and $j$-th coordinate.\\
Let $f:(S^n)^k \rightarrow S^n$ be a function, we denote by $f_{i,j}$ the function $f\circ \delta^{(k)}_{i,j}$.\\
\end{Def}

With these notations we deduce the following lemma:
\begin{Lem}\label{lem:fijanti}
If $f:(S^n)^k \rightarrow S^n$ satisfies the Twin Condition then  $f_{i,j}$ satisfies the Nowhere Anti-Unanimity Condition.
\end{Lem}

\begin{proof}
We consider the following functions 
\begin{eqnarray*}
\tilde{f}_{i,j}:(S^n)^2 &\rightarrow& S^n\\
(x,y)&\mapsto& f(e_1,\ldots,x,e_1,\ldots,e_1,y,e_1,\ldots,e_1)
\end{eqnarray*}
where $i \neq j $, $e_1=(1,0,\ldots,0) \in S^n \subset \RR^{n+1}$ and $x$ appears in $f_{i,j}(x)$ in the $i$-th coordinate and $y$ in the $j$-th coordinate.\\
As $f$ satisfies the Twin Condition then $\tilde{f}_{i,j}$ satisfies also the Twin Condition. Thus by Lemma  \ref{twinanti}, $\tilde{f}_{i,j}$ satisfies the Nowhere Anti-Unanimity Condition. Then for all $x \in S^n$ we have $f_{i,j}(x)=\tilde{f}_{i,j}(x,x) \neq -x$.
\end{proof}

The following lemma will be used to show that $f_{i,j}$ is homotopic to the identity function on the sphere $id_{S^n}$.

\begin{Lem}\label{lem:homotopy}
If a continuous function $g:S^n \rightarrow S^n$ satisfies the Nowhere Anti-Unanimity Condition then $g$ is  homotopic to $id_{S^n}$.
\end{Lem}
\begin{proof}
We consider $$h(t,x)=\dfrac{t.g(x)+(1-t). x}{\| t.g(x)+(1-t). x \|}.$$
The function $h$ is well defined and continuous because $g$ satisfies the Nowhere Anti-Unanimous Condition and then the denominator do not vanish. Furthermore, we have $h(0,x)=x=id_{S^n}(x)$ and $h(1,x)=g(x)$. Then $h$ is an homotopy between $g$ and $id_{S^n}$.
\end{proof}

The last ingredient for the proof of Theorem \ref{Thm:twin} is the following:
\begin{Lem}\label{lem:fin}
If $k \geq 3$ then there exists no continuous functions $f:(S^n)^k \rightarrow S^n$ such that for all $ i\neq j$, $f_{i,j}$ is homotopic to $id_{S^n}$.
\end{Lem}
The proof of this lemma uses classical homological properties. All these properties are recall in the appendix and  are denoted by $(P_1)$, $(P_2)$, $(P_3)$ and $(P_4)$.
\begin{proof}
By contradiction. Suppose that there exists a continuous function $f$ such that $f_{i,j}$ is homotopic to $id_{S^n}$, for all $i \neq j$.  The singular homology gives a functor $H_n$  which entails the following group homomorphism:
\begin{eqnarray*}
H_n(f):\ZZ^k & \longrightarrow & \ZZ\\
(z_1,\ldots,z_k)& \longmapsto & \sum_{\alpha=1}^k d_{\alpha} \times z_{\alpha}
\end{eqnarray*}
where $d_{\alpha} \in \ZZ$.\\
As $f_{i,j}=f\circ \delta^{(k)}_{i,j}$, by $(P_1)$ we have $$H_n(f_{i,j})=H_n(f) \circ H_n(\delta^{(k)}_{i,j}).$$ 
By $(P_2)$ we deduce that 
\begin{eqnarray*}
 H_n(f_{i,j}):\ZZ & \longrightarrow & \ZZ\\
 z &\longmapsto &(d_i+d_j)\times z
 \end{eqnarray*}
 Thus $$\deg(f_{i,j})=d_i+d_j,$$
 where $\deg$ denotes the degree of a continuous function from $S^n$ to $S^n$.\\
Furthermore, $f_{i,j}$ is homotopic to $id_{S^n}$ then by $(P_3)$, and $(P_4)$ we have 
$$\deg(f_{i,j})=1.$$ 
It follows 
 $$d_i+d_j=1 \textrm{ for all } i \neq j.$$
 As $k \geq 3$ this gives: $d_1+d_2=d_1+d_3=d_2+d_3=1$, and we deduce $d_1=1/2$. As $d_1$ must be an integer we get a contradiction and this proves the desired result.
\end{proof}

Now, we can prove Theorem \ref{Thm:twin}.
\begin{proof}[Proof of Theorem \ref{Thm:twin}]
By contradiction. Suppose that there exists a continuous function $f$ satisfying the Twin Condition. By Lemma \ref{lem:fijanti} and Lemma \ref{lem:homotopy}, for all $i \neq j$ we have $f_{i,j}$ homotopic to $id_{S^n}$. Then Lemma \ref{lem:fin} gives the desired contradiction.
\end{proof}

\section{The No Show Paradox}
In this section we describe  the No-Show Paradox in the topological social choice setting. In this situation we study the behavior of a family of aggregation rules. Indeed, in this situation the number of agents will change. Thus we study a family of aggregation rules $(f^{(k)})_k$ where the exponent $k$ corresponds to the number of agents. We would like to have an aggregation rule for which each player has an incentive to give his preference. This means  it is always better for players to give their preferences and to participate to the vote. This leads to the following definition:

\begin{Def} Let $(X,d)$ be a metric space and $(f^{(k)})_k$ be a family of functions from $X^k$ to $X$.\\
We say that the family  of functions $(f^{(k)})_k$ satisfies the Participation Condition  if for all $k \geq 2$ and for all $x_1,\ldots, x_{k+1} \in X$ we have:
$$d\Big(f^{(k+1)}(x_1,\ldots,x_i,\ldots,x_{k+1}),x_i\Big) \leq  d\Big(f^{(k)}(x_1,\ldots,\hat{x}_i,\ldots,x_{k+1}),x_i\Big),
$$
and the inequality is strict for $x_i \neq f^{(k)}(x_1,\ldots,\hat{x}_i,\ldots,x_{k+1})$.
\end{Def}
The notation $\hat{x}_i$ means that we omit $x_i$.\\

We remark that the Participation Condition implies 
$$f^{(k+1)}(x_1,\ldots,x_{i-1},f^{(k)}(x_1,\ldots,\hat{x}_i,\ldots,x_{k+1}),x_{i+1},\ldots,x_{k+1})=f^{(k)}(x_1,\ldots,\hat{x}_i,\ldots,x_{k+1}).$$
This means that  the Participation Condition entails that the ‘vote’ of an outsider who
agrees with the outcome will not change the outcome. \\
If a family of aggregation rule do not  satisfies the Participation Condition then we say that we have a No Show Paradox.\\

In this section we prove that we cannot avoid the No Show Paradox:

\begin{Thm}\label{Thm:noshow}
There exists no family  $f^{(k)}_k$ of continuous aggregation rules\\   $f^{(k)}:(S^n)^k \rightarrow S^n$ which satisfies the Participation Condition.
\end{Thm}

The proof of this theorem is similar to the one given for the Twin Paradox. We need to introduce the functions $f^{(k)}_{i,j}$ defined by $f^{(k)}_{i,j}=f^{(k)} \circ \delta^{(k)}_{i,j}$.

\begin{Lem}\label{lem:noshow}
If the family of functions $f^{(k)}:(S^n)^k \rightarrow S^n$  satisfies the Participation Condition then $f^{(k+1)}_{i,j}$ satisfies the Nowhere Anti-Unanimity Condition.
\end{Lem}

\begin{proof}
By contradiction. If there exists $x_0 \in S^n$ such that $f^{(k+1)}_{i,j}(x_0) = -x_0$ then 
 $$d_{S^n}\big(f^{(k+1)}_{i,j}(x_0),x_0\big)=d_{S^n}(-x_0,x_0)=\pi $$ 
 and the Participation Condition gives 
 \begin{eqnarray*}
d_{S^n}(f^{(k+1)}_{i,j}(x_0),x_0)&=& d_{S^n}\big(f^{(k+1)}(e_1,\ldots,e_1,x_0,e_1,\ldots,e_1,x_0,e_1\ldots,e_1),x_0\big) \\
&\leq& d_{S^n}\big(f^{(k)}(e_1,\ldots,e_1,x_0,e_1,\ldots,e_1),x_0\big).
\end{eqnarray*}
This implies $f^{(k)}(e_1,\ldots,e_1,x_0,e_1,\ldots,e_1)=-x_0$. \\
 Thus $x_0 \neq f^{(k)}(e_1,\ldots,e_1,x_0,e_1,\ldots,e_1)$ and in this situation the Participation Condition says that the previous inequality must be strict. Therefore,
$$\pi<d_{S^n}\big(f^{(k)}(e_1,\ldots,e_1,x_0,e_1,\ldots,e_1),x_0\big).$$ This gives a contradiction and thus proves $f^{(k+1)}_{i,j}(x) \neq  -x$, for all $x \in S^n$.\
\end{proof}

\begin{proof}[Proof of Theorem \ref{Thm:noshow}]
By contradiction. Suppose that there exists a family of continuous function satisfying the Participation Condition. By Lemma \ref{lem:noshow} and Lemma \ref{lem:homotopy}, for all $i\neq j$ we have $f^{(k+1)}$ homotopic to $id_{S^n}$. Then Lemma \ref{lem:fin} gives the desired contradiction.
\end{proof}
\section*{Appendix: An algebraic topology toolkit}

Our main tool is singular homology, see e.g. \citep{Hatcher,Rotman}. 
We recall briefly in this section some results from algebraic topology:\\

Singular homology gives a functor $H_n$ such that for a topological space $X$, $H_n(X)$ is an abelian group. Classical examples are $H_n(S^n)=\ZZ$ and $H_n\big((S^{n})^k\big)=\ZZ^k$. Furthermore for a continuous function  $f:X \rightarrow Y$, the functor $H_n$ gives a group homorphism $H_n(f):H_n(X) \rightarrow H_n(Y)$ and if $g:Y \rightarrow Z$ is a continous function we have the following property:
$$(P_1): \quad \quad H_n(f\circ g)=H_n(f)\circ H_n(g).$$

During our proofs we have used the continuous function $\delta^{(k)}_{i,j}$. The group homorphism $H_n(\delta^{(k)}_{i,j})$ associated to this continuous function is
\begin{eqnarray*}
(P_2): \quad \quad H_n(\delta^{(k)}_{i,j}): \ZZ & \longrightarrow & \ZZ^k\\
z&\longmapsto & (0,\ldots,0,z,0,\ldots,0,z,0,\ldots,0)
\end{eqnarray*}
where $z$ appears in $H_n(\delta^{(k)}_{i,j})(z)$ in the $i$-th and $j$-th coordinate.\\

Furthermore, if two continuous functions $f$ and $g$ are homotopic  then we have:
$$(P_3): \quad \quad H_n(f)=H_n(g).$$

A classical situation appears when we study a continuous function $f:S^n \rightarrow S^n$.\\
We recall that in this situation $H_n(f)$ is characterized by an \emph{integer} called the degree of $f$ and denoted by $\deg(f)$. The group homomorphism  $H_n(f):\ZZ \rightarrow \ZZ$ is then given by $z \mapsto \deg(f)\times z$.
The degree of the identity function on the sphere $S^n$ is equal to 1. We have then the property:
 $$(P_4): \quad \quad \deg(id_{S^n})=1.$$

In this note, we consider continuous aggregation rules on the sphere $S^n$. If one single point is removed from the domain $S^n$, i.e. the domain of individual preferences is reduced to $S^n \setminus \{x\}$ with $x$ in $S^n$, then our proof collapse. Indeed, we have $H_n(S^n \setminus \{x\})=0$.

\section*{Acknowledgement}
The author would like to thank the anonymous referees for their helpful comments.  The expression ``Nowhere Anti-Unanimity" is due to one of the referees.

 \bibliographystyle{kluwer} 
 
\bibliography{twin-paradox}

\end{document}